\documentclass[submission,copyright,creativecommons]{eptcs}
\providecommand{\event}{AFL 2026} % Name of the event you are submitting to

\usepackage{iftex}

\ifpdf
  \usepackage{underscore}         % Only needed if you use pdflatex.
  \usepackage[T1]{fontenc}        % Recommended with pdflatex
\else
  \usepackage{breakurl}           % Not needed if you use pdflatex only.
\fi

\title{The Value Generating Power of Weighted Tree Automata \\ with Initial Algebra Semantics}
\author{Manfred Droste
\institute{Leipzig University, Germany}
\and
Zolt\'an F\"ul\"op\footnote{Project no TKP2021-NVA-09 has been implemented with the support provided by the Ministry of Culture and Innovation of Hungary from the National Research, Development and Innovation Fund, financed under the TKP2021-NVA funding scheme.} 
\institute{University of Szeged, Hungary}
\and
Andreja Tepav\v{c}evi\'c 
\institute{Mathematical Institute SANU, Belgrade \\ University of Novi Sad, Serbia}
\and
Heiko Vogler
\institute{Technische Universit\"at Dresden, Germany}
}
\def\titlerunning{The Value Generating Power of Weighted Tree Automata}
\def\authorrunning{M. Droste,  Z. F\"ul\"op, A. Tepav\v{c}evi\'c, and H. Vogler}

\usepackage{amsmath,amssymb,dsfont,paralist,amsthm,eurosym}
\usepackage[format=default,labelfont=bf]{caption}
\usepackage{placeins}
\usepackage{microtype}
\usepackage[utf8]{inputenc}
\usepackage{scalerel}
\usepackage{stmaryrd}
\usepackage{bbold}
\usepackage[pagewise]{lineno}
\usepackage{algorithm}
\usepackage{algpseudocode}

\algnewcommand\algorithmicforeach{\textbf{for each}}
\algnewcommand\algorithmicvariables{\textbf{Variables:}}
\algnewcommand\Variables{\item[\algorithmicvariables]}
\algdef{S}[FOR]{ForEach}[1]{\algorithmicforeach\ #1\ \algorithmicdo}
\algrenewcommand\algorithmicrequire{\textbf{Input:}}
\algrenewcommand\algorithmicensure{\textbf{Output:}}
\algnewcommand\Fixedcomment[1]{\hfill\makebox[0.4\textwidth][l]{$\triangleright$ #1}}

\usepackage{graphicx}
\usepackage{tikz}
\usepackage{subcaption}
\usepackage{tikz-qtree}
\usepackage{rotating}

\usetikzlibrary{arrows,matrix}
\usetikzlibrary{matrix,calc}
\usetikzlibrary{shapes.geometric}
\usetikzlibrary{positioning}
\usetikzlibrary{calligraphy}
\usetikzlibrary{fit}

\pgfdeclarelayer{background}
\pgfsetlayers{background,main}

\usetikzlibrary{decorations.pathreplacing}
\tikzset{curlybrace/.style={decoration=brace,decorate}}
\usetikzlibrary{automata}
\usetikzlibrary{calc,positioning}
\usetikzlibrary{decorations.pathmorphing}
\usetikzlibrary{shapes}
\tikzset{trinode/.style={draw,triangle,minimum width=2.0cm}}
\tikzset{snake/.style={decorate,decoration=snake}}
\tikzset{curlybrace/.style={decoration=brace,decorate}}
\tikzset{triangle/.style={regular polygon,regular polygon sides=3}}
\tikzset{edge from parent path={(\tikzparentnode) -- (\tikzchildnode.north)}}
\usetikzlibrary{arrows.meta}

\usepackage{wrapfig}
\usepackage[rightcaption]{sidecap}
\usepackage[colorinlistoftodos]{todonotes}
\usepackage{mathtools}
\usepackage{enumitem}
\usepackage{multicol} 
\usepackage{todonotes}
\usepackage{imakeidx}         % allows index generation

\usepackage{hyperref}

\newtheorem{theorem}{Theorem}[section]
\newtheorem{main-theorem}{Main Theorem}
\newtheorem{lemma}[theorem]{Lemma}

\newtheorem{corollary}[theorem]{Corollary}
\newtheorem{observation}[theorem]{Observation}
\newtheorem{example}[theorem]{Example}
\newtheorem{definition}[theorem]{Definition}

\newcommand{\cA}{\mathcal A}
\newcommand{\cB}{\mathcal B}

\newcommand{\rmM}{\mathrm{M}}

\newcommand{\C}{\mathrm{C}}

\newcommand{\T}{\mathrm{T}}

\newcommand{\R}{\mathrm{R}}

\newcommand{\im}{\mathrm{im}}

\newcommand{\wt}{\mathrm{wt}}

\newcommand{\e}{\mathrm{e}} % for elementary context

\newcommand{\sfM}{\mathsf{M}}
\newcommand{\sfA}{\mathsf{A}}
\newcommand{\sfB}{\mathsf{B}}

\newcommand{\A}{\mathsf{A}}
\newcommand{\B}{\mathsf{B}}

\newcommand{\M}{\mathsf{M}}

\newcommand{\sfT}{\mathsf{T}}

\newcommand{\V}{\mathsf{V}}

\newcommand{\0}{\mathbb{0}}
\newcommand{\1}{\mathbb{1}}

\newcommand{\h}{\mathrm{h}}

\newcommand{\pos}{\mathrm{pos}}

\newcommand{\rk}{\mathrm{rk}}

\newcommand{\runsem}[1]{[\![#1]\!]^{\mathrm{run}}}
\newcommand{\initialsem}[1]{[\![#1]\!]^{\mathrm{init}}}

\newcommand{\id}{\mathrm{id}}

\newcommand{\rmST}{{\mathrm{ST}_\Sigma}}
\newcommand{\sfST}{\mathsf{ST}}
\newcommand{\rmS}{\mathrm{S}}
\newcommand{\sfS}{\mathsf{S}}

\pgfdeclaredecoration{dashsoliddouble}{initial}{
  \state{initial}[width=\pgfdecoratedinputsegmentlength]{
    \pgfmathsetlengthmacro\lw{.7pt+.5\pgflinewidth}
    \begin{pgfscope}
      \pgfpathmoveto{\pgfpoint{0pt}{\lw}}%
      \pgfpathlineto{\pgfpoint{\pgfdecoratedinputsegmentlength}{\lw}}%
      \pgfmathtruncatemacro\dashnum{%
        round((\pgfdecoratedinputsegmentlength-3pt)/6pt)
      }
      \pgfmathsetmacro\dashscale{%
        \pgfdecoratedinputsegmentlength/(\dashnum*6pt + 3pt)
      }
      \pgfmathsetlengthmacro\dashunit{3pt*\dashscale}
      \pgfsetdash{{\dashunit}{\dashunit}}{0pt}
      \pgfusepath{stroke}
      \pgfsetdash{}{0pt}
      \pgfpathmoveto{\pgfpoint{0pt}{-\lw}}%
      \pgfpathlineto{\pgfpoint{\pgfdecoratedinputsegmentlength}{-\lw}}%     
      \pgfusepath{stroke}
    \end{pgfscope}
  }
}

\tikzset{small circle/.style={circle, draw=black, inner sep=0pt,outer sep=0pt, minimum size=2.5pt}}

\begin{document}
\maketitle

\begin{abstract}
We consider the generating power of the initial algebra semantics of weighted tree automata over strong bimonoids (hence also over semirings) and the question under which conditions the weighted tree automata can produce only finitely many values. 
We show that there exists a right-distributive strong bimonoid which is bi-locally finite but not locally finite. We also show that if the ranked alphabet contains a symbol with rank at least two,
then for any finitely generated strong bimonoid, weighted tree automata
can generate, via their initial algebra semantics, all elements of the strong bimonoid.
As a consequence of these results, for bi-locally finite right-distributive strong bimonoids which are not locally finite,
weighted tree automata can generate infinitely many values, provided that the input ranked alphabet
contains a symbol with rank at least two. This is in sharp contrast to the setting of weighted string automata,
which can generate only finitely many values.
As a further consequence, for any finitely generated semiring, there exists
a weighted tree automaton which generates, via its run semantics, all elements of  the semiring.
\end{abstract}

\section{Introduction}

Trees are a fundamental data structure in theoretical computer science,
naturally modeling hierarchical information such as syntax in programming languages,
structured data formats, and derivation trees in logic and linguistics.
Correspondingly, tree automata provide a robust and well-studied framework for recognizing and processing tree-structured inputs,
extending classical automata theory from strings to branching objects.
In many applications, one also needs to consider quantitative information such as costs, probabilities, or multiplicities.
This has led to the development of weighted tree automata,
which combine the structural capabilities of tree automata with algebraic weight domains, typically given by semirings.
Weighted tree automata have found applications in areas such as natural language processing and program analysis,
cf., e.g., \cite{knigra05,rephorsag95,chi05}; the rich theory of weighted automata, started with \cite{sch61}, is witnessed by the books and surveys \cite{eil74,salsoi78,wec78,kuisal86,berreu88,sak09,drokuivog09,drokus21,fulvog26}.

A crucial property for weighted string and tree automata over semirings is that two descriptions of their behavior,
the \emph{run semantics} and the \emph{initial algebra semantics},
coincide. This enables the application of methods from algebra for the solution of combinatorial
and language-theoretic problems of weighted automata and even of classical unweighted automata.
More recently, strong bimonoids have been investigated as weight structures for such automata \cite{drostuvog10,drovog10,rad10,cirdroignvog10,drofulkosvog21}.
Strong bimonoids are structures  $\B = (B,\oplus,\otimes,\0,\1)$ where  $(B,\oplus,\0)$  is a commutative monoid,
$(B,\otimes,\1)$  is a monoid, and  $\0$  acts like a multiplicative zero.
Here, we will investigate weighted automata over \emph{right-distributive} strong bimonoids;
these can be viewed as semirings without requiring left-distributivity. Typical examples include, e.g., the standard real-valued polynomials with usual addition and with composition
as multiplication operation (which is right-, but not left-distributive).
It has been shown that the important coincidence of run and initial algebra semantics
extends to weighted string automata over right-distributive strong bimonoids \cite{drostuvog10}. 
For many further results and examples for right-distributive strong bimonoids, see \cite{fulvog26}.

A central topic in the theory of weighted automata is the description of their generating power 
for outputs, including conditions under which weighted automata produce only finitely many values.
If the right-distributive strong bimonoid  $\B$  is \emph{bi-locally finite}, i.e., both monoids  $(B,\oplus,\0)$
and  $(B,\otimes,\1)$  are locally finite, then the run semantics, and hence also its initial algebra semantics,
of any weighted string automaton clearly is finite-valued.

It is the goal of this paper to investigate whether analogous results hold true for weighted tree automata.
It is easy to see that if  $\B$  is bi-locally finite, the run semantics of a weighted tree automaton is finite-valued.
But, it is known that for right-distributive strong bimonoids  $\B$, the run semantics and the initial algebra semantics
of weighted tree automata, in contrast to weighted string automata, may differ (cf. e.g. \cite{rad10,fulvog26}).
Hence for bi-locally finite right-distributive strong bimonoids the question arises
whether the initial algebra semantics of a weighted tree automaton is also finite-valued.
If  $\B$ is \emph{locally finite}, i.e., each finitely generated strong submonoid of  $\B$ is finite,
then clearly the initial algebra semantics of any weighted tree automaton is finite-valued.
Also, if  $\B$  is locally finite, trivially  $\B$  is bi-locally finite.
Hence an important question is whether in general for right-distributive strong bimonoids 
local finiteness and bi-local finiteness are equivalent.
Our first main result will be:

\begin{main-theorem}\label{thm:main-result-on-B}\rm  There exists a bi-locally finite right-distributive strong bimonoid  $\B$  which is not locally finite.
\end{main-theorem}

We note that right-distributive strong bimonoids are near-semirings,
which have been investigated in algebra, cf. \cite{hooroo67,pil83};
hence Main Theorem \ref{thm:main-result-on-B} also shows that there are bi-locally finite near-semirings which are not locally finite.
We will prove a slightly stronger version of Main Theorem \ref{thm:main-result-on-B} in Theorem \ref{thm:M-is-what-we-want}.
In our construction, we start with the free algebra of terms, with addition and multiplication as operations,
over a set  $X$.  By a natural congruence, the quotient of the term algebra
becomes the free right-distributive strong bimonoid over  $X$.
We construct two suitable further congruence conditions 
which ensure that the resulting quotient strong bimonoid  $\B$
is bi-locally finite, but, using the freeness of the construction, we can show that
$\B$  is not locally finite.

Now we turn to the question whether for all weighted tree automata  $\cA$  over a bi-locally finite right-distributive strong bimonoid,
the image $\im(\initialsem{\cA})$  of its initial algebra semantics is finite.
Our second main result will show (together with the above first main result)
that this statement fails drastically as soon as the ranked alphabet contains a symbol of rank at least two.
In fact, for any finitely generated strong bimonoid  $\B$  there exist weighted tree automata 
over such a ranked alphabet and  $\B$  which have the universality property:

\begin{main-theorem}\rm \label{thm:main-result-on-Sigma} Let  $\Sigma$  be an arbitrary ranked alphabet containing a symbol with rank at least 2, i.e., $\Sigma \ne \Sigma^{(0)} \cup \Sigma^{(1)}$, and let  $\B$  be any finitely generated strong bimonoid.
Then there exists a weighted tree automaton $\cA$  over  $\Sigma$ and  $\B$  such that $\im(\initialsem{\cA}) = B$.
\end{main-theorem}

Note that the values of $\im(\initialsem{\cA})$  are given by iterated sums and products of the
finitely many weights occurring in the weighted tree automaton $\cA$.
The constructed weighted tree automaton  $\cA$  produces, at proper branching points,
also \emph{products of sums of weights} occurring in  $\cA$
and, continuing in this way, generates \emph{all} values of the underlying strong bimonoid  $\B$
(provided it satisfies the necessary assumption of being finitely generated).
In particular, trivially, if  $\B$  is infinite, the image  $\im(\initialsem{\cA})$  is also infinite.
By Main Theorem~\ref{thm:main-result-on-B}, this may also happen if  $\B$  is bi-locally finite and right-distributive, answering the question stated above and
showing that weighted tree automata over ranked alphabets containing a symbol with rank a least 2
have much larger generating power than weighted string automata. Main Theorem~\ref{thm:main-result-on-Sigma} will be proved as Theorem \ref{lm:closure-of-finite-set-i-recognizable-stronger}.

To the best of our knowledge, Main Theorem~\ref{thm:main-result-on-Sigma} is new also for the case of semirings.
As an immediate consequence of Main Theorem~\ref{thm:main-result-on-Sigma} and the coincidence of the run semantics and the initial algebra semantics for semirings, we obtain that for each ranked alphabet
$\Sigma$  containing a symbol with rank at least two and for each finitely generated semiring  $\sfB$,
there exists a weighted tree automaton  $\cA$  over  $\Sigma$  and  $\sfB$  such that  $\im(\runsem{\cA}) = B$.
Note that we obtained this result as a consequence of the construction for the initial algebra semantics.
It would be interesting to see whether a direct proof for the run semantics is possible.

Right-distributive strong bimonoids provide very natural generalizations of semirings.
As mentioned, they are particular near-semirings studied in algebra.
Polynomials and endomorphisms of monoids with pointwise addition and composition as multiplication
form natural examples of near-semirings and right-distributive strong bimonoids.
The usual multiplication of ordinal numbers is left- but not right-distributive over addition.
In computer science, in compositional semantics and interprocedural dataflow analysis, 
trees represent structured computations where transfer functions are combined across branches 
and composed along evaluation paths, with joins aggregating alternative behaviors. 
This induces a non-distributive algebra of transformations with join as addition and composition as multiplication, 
naturally captured by weighted tree automata over near-semirings and right-distributive strong bimonoids, cf. \cite{shapnu81,rephorsag95}.

Open research questions are presented and discussed in Section \ref{sect:further-research}. A preliminary version of this paper is published as \cite{drofultepvog24}.

\section{Preliminaries}
\label{sec:preliminaries}

\noindent {\bf General.} We denote by $\mathbb{N}$  the set of nonnegative integers and $\mathbb{N}_+=\mathbb{N}\setminus\{0\}$. For every $k,n \in \mathbb{N}$, we denote by $[k,n]$ the set $\{i\in \mathbb{N} \mid k \le i \le n\}$ and we abbreviate $[1,n]$ by $[n]$. Hence $[0] = \emptyset$.

Let $A$ be a set.  We denote by $A^*$ the set of all finite sequences over $A$, and by $\varepsilon$ the empty sequence. 
Let $\rho \subseteq A\times A$ be a binary relation on $A$. As usual, $\rho^*$ denotes the reflexive and transitive closure of $\rho$.   If $\rho$ is
an equivalence relation, then for each $a\in A$, the \emph{equivalence class of $a$ (modulo $\rho$)}, denoted by $[a]_\rho$, is the set $\{b\in A \mid a\rho b\}$. For each $B\subseteq A$, we put $B/_\rho =\{[a]_\rho \mid a\in B\}$.

\noindent {\bf Ranked alphabets and terms.}    A {\em ranked alphabet} is a pair $(\Sigma,\rk)$, where $\Sigma$ is a non-empty and finite set and $\rk: \Sigma \rightarrow \mathbb{N}$ is a mapping, called \emph{rank  mapping},  such that $\rk^{-1}(0)\not= \emptyset$. For each $k\in \mathbb{N}$, we denote the set $\rk^{-1}(k)$ by $\Sigma^{(k)}$. Sometimes we write $\sigma^{(k)}$ to indicate that $\sigma \in \Sigma^{(k)}$. Moreover, we abbreviate $(\Sigma,\rk)$ by $\Sigma$ and assume that the rank mapping is known or irrelevant.
We call $\Sigma$ a \emph{string ranked alphabet} if $\Sigma=\Sigma^{(0)} \cup \Sigma^{(1)}$ with $|\Sigma^{(0)}|=1$ and $\Sigma^{(1)}\not= \emptyset$.
Let $\Sigma$ be a ranked alphabet and $X$ be a set disjoint with $\Sigma$.
The set of \emph{$\Sigma$-terms (or $\Sigma$-trees) over $X$, } denoted by $\T_\Sigma(X)$, is the smallest set $T$ such that (i) $\Sigma^{(0)} \cup X \subseteq T$ and (ii) for every $k \in \mathbb{N}_+$, $\sigma \in \Sigma^{(k)}$, and $t_1,\ldots, t_k \in T$, we have $\sigma(t_1,\ldots,t_k) \in T$. We write $\T_\Sigma$ for $\T_\Sigma(\emptyset)$.

\noindent{\bf Universal algebra.} 
  We recall some concepts and results from universal algebra (cf. \cite{bursan81,wec92,baanip98}).
In universal algebra, a ranked alphabet is called a signature. Each letter of the ranked alphabet is an operation symbol of the same arity. Nullary letters (leaves) are nullary operations (or constants). 

Let $\Sigma$ be a ranked alphabet (signature). An \emph{algebra $\sf A$ of type $\Sigma$} ($\Sigma$-algebra) is a pair $\sfA=(A,\theta)$ where $A$ is a non-empty set and $\theta$ is a mapping from $\Sigma$ to the family of finitary operations on $A$ such that for every $k\in \mathbb{N}$ and  $\sigma \in \Sigma^{(k)}$, the arity of the operation $\theta(\sigma)$ is $k$. In particular, for $k=0$, a nullary operation is a mapping of type $A^0 \to A$ where $A^0=\{()\}$. 

Let $X$ be a set. The \emph{$\Sigma$-term algebra over $X$}, denoted by $\sfT_\Sigma(X)$,  is the $\Sigma$-algebra
\(\sfT_\Sigma(X) = (\T_\Sigma(X),\theta_\Sigma)\) where, for every $k \in \mathbb{N}$, $\sigma \in \Sigma^{(k)}$, and $t_1,\ldots, t_k \in \T_\Sigma(X)$, we let $\theta_\Sigma(\sigma)(t_1,\ldots,t_k) = \sigma(t_1,\ldots,t_k)$. The \emph{$\Sigma$-term algebra}, denoted by $\sfT_\Sigma$, is the $\Sigma$-term algebra over $\emptyset$, i.e., $\sfT_\Sigma = \sfT_\Sigma(\emptyset)$.

Let $\sfA=(A,\theta)$ be a $\Sigma$-algebra and $A' \subseteq A$. We say that \emph{$A'$ is closed under $\theta(\Sigma)$} if, for every $k \in \mathbb{N}$,  $\sigma \in \Sigma^{(k)}$, and $a_1,\ldots,a_k \in A'$, we have $\theta(\sigma)(a_1,\ldots,a_k) \in A'$. We denote by $\langle A' \rangle_{\theta(\Sigma)}$ the smallest subset of $A$ which contains $A'$ and is closed under $\theta(\Sigma)$. The \emph{subalgebra of $\sfA$ generated by $A'$} is the $\Sigma$-algebra $( \langle A' \rangle_{\theta(\Sigma)},\theta')$ where $\theta'$ is obtained from $\theta$ by restricting each operation $\theta(\sigma)$ to $\langle A' \rangle_{\theta(\Sigma)}$.
If $A=\langle A' \rangle_{\theta(\Sigma)}$, then $\sfA$ is \emph{generated by $A'$}.
The $\Sigma$-algebra $\sfA=(A,\theta)$ is \emph{locally finite} if, for each finite subset $A' \subseteq A$, the set $\langle A' \rangle_{\theta(\Sigma)}$ is finite.

Let $\sfA_1=(A_1,\theta_1)$ and $\sfA_2=(A_2,\theta_2)$ be $\Sigma$-algebras. A \emph{$\Sigma$-algebra homomorphism (from $\sfA_1$ to $\sfA_2$)}  is a mapping $h: A_1 \to A_2$ such that the equality $h(\theta_1(\sigma)(a_1,\ldots,a_k)) = \theta_2(\sigma)(h(a_1),\ldots,h(a_k))$ holds for every $k \in \mathbb{N}$, $\sigma \in \Sigma^{(k)}$, and $a_1,\ldots,a_k \in A_1$.

In the rest of this section, $\sfA$ denotes an arbitrary $\Sigma$-algebra $(A,\theta)$.
  
Let $\mathcal{K}$ be an arbitrary class of $\Sigma$-algebras.
An algebra  $\A$  in  $\mathcal{K}$  is called  \emph{initial in} $\mathcal{K}$,
if for each algebra ~$\A'$ in $\mathcal{K}$, there exists exactly one $\Sigma$-algebra homomorphism $h\colon A \rightarrow A'$  from $\A$ to $\A'$ 
(cf. \cite[p.~164, Def.~4]{wec92}).
The $\Sigma$-term algebra $\sfT_\Sigma$ is initial in the class of all $\Sigma$-algebras.

A \emph{congruence on $\sfA$} is an equivalence relation $\rho\subseteq A\times A$ satisfying the following condition: for every $k\in \mathbb{N}$, $\sigma \in \Sigma^{(k)}$, $a_1,b_1,\ldots,a_k,b_k \in A$, if $a_i \,\rho\, b_i$ for each $i\in[k]$, then
\(\theta(\sigma)(a_1,\ldots,a_k ) \, \rho \, \theta(\sigma)(b_1,\ldots,b_k )\). 
Let  $\rho$ be a congruence on $\sfA$. The \emph{quotient algebra of $\sfA$ by $\rho$} is the $\Sigma$-algebra $\sfA/\!_\rho = (A/\!_\rho, \theta/\!_\rho)$ where, for every $k \in \mathbb{N}$,  $\sigma \in \Sigma^{(k)}$, and $[a_1]_{\rho},\ldots,[a_k]_{\rho} \in A/\!_\rho$, we have $\theta/\!_\rho(\sigma)([a_1]_{\rho},\ldots,[a_k]_{\rho}) =  [\theta(\sigma)(a_1,\ldots, a_k)]_{\rho}$.

Next we wish to consider $\Sigma$-identities and the congruence on $\sfA$ induced by a set of such identities.  
Let $Z=\{z_1,z_2,\ldots\}$ be a set of variables. For each $n\in \mathbb{N}$, we put $Z_n=\{z_1,\ldots,z_n\}$.  
Each element $t\in \T_\Sigma(A\cup Z_1)$ in which $z_1$ occurs exactly once is called a {\em $\Sigma A$-context}. The set of all $\Sigma A$-contexts is denoted by $\C_{\Sigma,A}$ (cf. unary algebraic functions in  \cite[Def.~1.3.13]{gecste84}). 
Let $t\in \T_\Sigma(A\cup Z_n)$ for some $n\in \mathbb{N}$. The mapping $t^\sfA: A^n \to A$ is defined by structural induction in a standard way. Such a mapping is called 
algebraic function in \cite[p.~22]{gecste84} and term function in \cite[Def.~II.10.2]{bursan81} for the special case that $t\in \T_\Sigma(Z_n)$.

An \emph{assignment} is a mapping $\varphi: Z \to A$. Each such mapping $\varphi$ extends uniquely to a $\Sigma$-algebra homomorphism
$\varphi'$ from  $\sfT_\Sigma(A\cup Z)$ to $\sfA$ satisfying that $\varphi'(a)=a$ for each $a\in A$. In the sequel, we drop the prime from  $\varphi'$. For an arbitrary $t \in \T_\Sigma(A\cup Z)$, we call $\varphi(t)$ the {\em evaluation of $t$ in $ \sfA$ at $\varphi$}. 
We note that for each $n\in \mathbb{N}$, $t\in \T_\Sigma(A\cup Z_n)$, and assignment $\varphi$ with $\varphi(z_i)=a_i$ for $i\in [n]$, we have $\varphi(t)=t^{\sfA}(a_1,\ldots,a_n)$.

A \emph{$\Sigma$-identity over $Z$} (or: identity) is a pair $(\ell, r)$ where $\ell,r \in \T_\Sigma(Z)$.  The $\Sigma$-algebra $\sfA$ \emph{satisfies the identity $(\ell, r)$} if, for every assignment $\varphi: Z \to A$, we have
$\varphi(\ell) = \varphi(r)$.

\begin{lemma}\rm \label{identitiestofactoralgebras}\cite[Th.~II.6.10 and Lm.~II.11.3]{bursan81}
If $\sfA$ satisfies an identity $(\ell, r)$ and $\rho$ is a congruence on $\sfA$, then $\sfA/\!_\rho$ also satisfies the identity $(\ell, r)$. 
\end{lemma}

Let $E$ be a set of identities. The \emph{congruence (relation)  on $\sfA$ induced by $E$}, denoted by $\approx_E$, is the smallest congruence on $\sfA$ which contains the set
\begin{equation}\label{eq:identity-in-term-algebra}
E(\sfA)=\{(\varphi(\ell), \varphi(r))\mid (\ell, r)\in E, \varphi: Z \to A\}.
\end{equation}

The following lemma is well known and can be proven similarly to \cite[p.~176, Lm.~24]{wec92}.

\begin{lemma}\label{lm:free-algebra-quotient} \rm Let  $E$ be a set of $\Sigma$-identities. Then $\sfA/\!_{\approx_E}$ satisfies all identities in $E$.
\end{lemma}

Next we extend the well-known syntactic characterization of the congruence on $\sfT_\Sigma(Z)$ induced by a set $E \subseteq \T_\Sigma(Z) \times \T_\Sigma(Z)$ of $\Sigma$-identities, 
cf. \cite[Thm.~ 3.1.12]{baanip98} and \cite[Thm.~II.14.17, II.14.19]{bursan81},
to a characterization of the congruence on $\sfA$ induced by  $E$. In fact, this is closely related to a general description of a congruence generated by a binary relation on  $\A$,
cf. \cite[Sect.~2.1.2]{wec92}.

Let $E$ be a set of $\Sigma$-identities. The \emph{reduction relation induced by $E$ on $A$}, denoted  by  $\Rightarrow_E$, is the binary relation on $A$ defined as follows:
for every $a,b \in A$, we let $a \Rightarrow_E b$ if there exist a $\Sigma A$-context $c\in \C_{\Sigma,A}$, an identity $(\ell, r)$ in $E$,  and an assignment $\varphi: Z \to A$ such that
$a= c^\sfA(\varphi(\ell))$ and $b = c^\sfA(\varphi(r))$.
In this  case  we say that $b$ is obtained from $a$  in a reduction step (using the identity $(\ell, r)$).
For an identity $e=(\ell,r)$ we define $e^{-1}=(r,\ell)$ and we let $E^{-1}=\{ e^{-1} \mid e\in E\}$.
Moreover, we abbreviate $\Rightarrow_{E\cup E^{-1}}$ by $\Leftrightarrow_E$.

The subsequent characterization says that, for any two elements $a,b\in A$,
we have $a \approx_E b$ if and only if, there is a finite sequence of elements $a=a_0,a_1,\ldots,a_n=b$ of $A$ for some $n\in \mathbb{N}$ such that for each $i\in [n]$, the element $a_i$ can be obtained from $a_{i-1}$ in a reduction step using an identity in $E$  or the inverse of an identity. 
%We include the proof in the Appendix. 

\begin{lemma}\rm \label{lm:approx-characterization}\  \cite[p.~98, Thm.~6]{wec92} Let  $E$ be a set of $\Sigma$-identities and $\approx_E$ the congruence on $\sfA$ induced by $E$. Then $\approx_E \, =\, \Leftrightarrow^*_E$.
\end{lemma}

\noindent {\bf Strong bimonoids.}
\sloppy A \emph{strong bimonoid} \cite{drostuvog10,cirdroignvog10,rad10,drovog10,drovog12} is an algebra $\B=(B,\oplus,\otimes,\0,\1)$ such that $(B,\oplus,\0)$ is a commutative monoid, $(B,\otimes,\1)$ is a monoid, and $\0$ is annihilating with respect to $\otimes$, i.e., for each $b \in B$ we have $b\otimes \0 = \0 \otimes b = \0$. The operations $\oplus$ and $\otimes$ are called addition and multiplication, respectively. For examples of strong bimonoids we refer to \cite{drostuvog10,cirdroignvog10} (also cf. \cite[Ex.~2.7.15]{fulvog26}).

Let $\B=(B,\oplus,\otimes,\0,\1)$ be a strong bimonoid. It is
\begin{compactitem}
\item \emph{idempotent} if, for each $b \in B$, we have $b \oplus b = b$,
\item  \emph{almost idempotent} if, for each $b \in B$, we have $ b \oplus b = b \oplus b \oplus b$,
\item \emph{left-distributive} if, for every $a,b,c \in B$, we have $a \otimes (b \oplus c) = (a \otimes b) \oplus (a \otimes c)$,
\item \emph{right-distributive} if, for every $a,b,c \in B$, we have $(a \oplus b) \otimes c = (a \otimes c) \oplus (b \otimes c)$,
    \item \emph{additively locally finite} if  $(B,\oplus,\0)$ is locally finite,
  \item \emph{multiplicatively locally finite} if  $(B,\otimes,\1)$ is locally finite,
  and
  \item \emph{bi-locally finite} if it is additively and multiplicatively locally finite.
  \end{compactitem}
A \emph{semiring} \cite{hebwei93,gol99} is a distributive (i.e., left-distributive and right-distributive)  strong bimonoid.

\begin{observation}\label{obs:biloc-fin+right-disrt-loc-fin} \rm   Let $\B=(B,\oplus,\otimes,\0,\1)$ be a strong bimonoid.  
  \begin{compactitem}
  \item[(a)] If $\B$ is locally finite,  then it is bi-locally finite.
  \item[(b)] If $\B$ is almost idempotent, then it is additively locally finite.
      \end{compactitem}
\end{observation}

\begin{example} \rm \label{ex:maps-N}
 We consider the set $C$ of mappings $f: \mathbb{N} \to \mathbb{N}$ such that $f(0)=0$. Then the algebra $\mathrm{Maps}(\mathbb{N})=(C,+,\circ,\widetilde{0},\id)$ is a strong bimonoid where
for every $f_1,f_2: \mathbb{N} \to \mathbb{N}$ and $n \in \mathbb{N}$, we let $(f_1 + f_2)(n) = f_1(n) + f_2(n)$ and $(f_1 \circ f_2)(n)=f_1(f_2(n))$ and
 $\widetilde{0}(n) = 0$ and $\id(n)=n$. 
  Obviously, $\mathrm{Maps}(\mathbb{N})$ is right-distributive. However, it is not left-distributive, because e.g.
 \((\mathrm{sq} \circ (\id + \id ))(1) = 4 \ne 2 = (\mathrm{sq} \circ \id + \mathrm{sq} \circ \id)(1)
  \)
  where $\mathrm{sq}: \mathbb{N} \to \mathbb{N}$ be such that $\mathrm{sq}(n) = n^2$  for each $n \in \mathbb{N}$.
  \hfill $\Box$
    \end{example}

%%%%%%%%%%%%%%%%%%%%%%%%%%%%%%%%%%%%%%%%%%%%%

\section{Bi-locally  finite and right-distributive strong bimonoids}\label{sect:wlc-strong-bimonoids}

The main goal of this section is to show that there exists a right-distributive strong bimonoid which is bi-locally finite but not locally finite, cf.~Theorem~\ref{thm:M-is-what-we-want}.
\emph{In the rest of this section, we let $\Sigma = \{+^{(2)}, \times^{(2)}, 0^{(0)}, 1^{(0)}\}$ and let $X$ be a non-empty set.}

For the proof, we define an almost idempotent right-distributive strong bimonoid $\sfM(X)$ and prove that it is bi-locally finite  and not locally finite. We define $\M(X)$ in two steps. In the first step, we define an algebra  $\sfST_\Sigma(X)$ which looks very similar to the free  $\Sigma$-algebra  $\sfT_\Sigma(X)$, but incorporates
the usual laws for  $0$  and $1$. Then the quotient
$\sfS(X) = \sfST_\Sigma(X)/_{\approx_E}$  by the congruence induced by a set  $E$ of natural identities is a strong bimonoid which is right-distributive and almost idempotent (see Lemma
\ref{prop:SB-strong-bimonoid}). The almost idempotency implies that  $\sfS(X)$  is already additively locally finite.
In the second step, we factorize $\sfS(X)$ by the congruence relation $\sim_\mathrm{la}$ which identifies multiplicatively ``large'' elements of $\sfS(X)$ (cf. Definition~\ref{def:large-polynomial}) in order to obtain the multiplicatively locally finite algebra $\M(X)$. In Theorem~\ref{thm:M-is-what-we-want} we prove that $\M(X)$ is not locally finite.
  This uses Lemma \ref{lm:approx-characterization} and exploits
our choice of the identities of  $E$, namely, that particular
constructed terms permit only one reduction rule (or its inverse) from  $E$.

\underline{Step 1:} We write elements of $\T_\Sigma(X)$ in infix form, e.g., we write $(1+1)\times x$ for $\times(+(1,1),x)$ where $x\in X$.
We call a term  $t  \in \T_\Sigma(X)$  \emph{simple}, if  
\begin{compactitem}
\item $t = 0$  or
\item $t\ne 0$  and it contains neither   $0$, nor a subterm of the form  $1 \times s$,  nor a subterm of the form $s \times 1$.
\end{compactitem}Let $\rmST(X)$ denote the set of all simple terms in $\T_\Sigma(X)$. Note that $1$ is simple, hence
e.g., $1+t\in\rmST(X) $ for each $t\in \rmST(X)$. 
We define the algebra $\sfST_\Sigma(X)=(\rmST(X),+_\mathsf{ST},\times_\mathsf{ST},0,1)$ as follows:
\begin{compactitem}
\item for each $t\in \rmST(X)$, let $t+_\mathsf{ST} 0= 0 +_\mathsf{ST} t=t$ and   $t \times_\mathsf{ST} 0= 0 \times_\mathsf{ST} t=0$,
\item for each $t\in \rmST(X)$, let $t\times_\mathsf{ST} 1= 1 \times_\mathsf{ST} t=t$,
\item for every $s,t\in \rmST(X)\setminus\{0,1\}$, let $s +_\mathsf{ST} t = s+t$ and $s \times_\mathsf{ST} t = s\times t$.
\end{compactitem}

We note here that $\sfST_\Sigma(X)$  is the free algebra over  $X$  in the class of all $\Sigma$-algebras  $\mathsf{A}=(A,\theta)$  satisfying the laws  $a + 0 = 0 + a = a$, $a \times 0 = 0 \times a = 0$  and  $a \times 1 = 1 \times a = a$  for each  $a \in A$, cf. e.g. \cite[p.~164, Def.~4]{wec92}.

Next let $E = \{e_1,e_2,e_3,e_4,e_5\}$ be the set of the following five identities:
\begin{center}
\begin{tabular}{ll}
$e_1: \big(z_1 + (z_2 + z_3) \ , \ (z_1 + z_2) + z_3\big)$ \hspace{8mm} & $e_4: \big(z_1 + z_1 \ , \ z_1 + (z_1 + z_1) \big)$ \\[2mm]
$e_2: \big(z_1 + z_2 \ , \ z_2 + z_1\big)$ & $e_5:  \big((z_1 + z_2)\times z_3 \ , \ (z_1 \times z_3) + (z_2 \times z_3) \big)$ \\[2mm]
$e_3: \big(z_1 \times (z_2 \times z_3)  \ , \  (z_1 \times z_2) \times z_3\big)$  
\end{tabular}
\end{center}
We note that $E \subseteq \T_\Sigma(Z_3) \times \T_\Sigma(Z_3)$.
Then we consider the quotient algebra 
\[\sfST_\Sigma(X)/_{\approx_E}=(\rmST(X)/_{\approx_E},+_\mathsf{ST}/_{\approx_E},\times_\mathsf{ST}/_{\approx_E},[0]_{\approx_E},[1]_{\approx_E}).\]
Lastly, we abbreviate the latter notation by $\sfS(X)=(\rmS(X),+_\sfS,\times_\sfS,0_\sfS,1_\sfS)$, and,  for each $s\in \rmST(X)$, we abbreviate $[s]_{\approx_{E}}$ by $[s]_{E}$.

\begin{lemma}\rm \label{prop:SB-strong-bimonoid} The algebra $\sfS(X)$ is an almost idempotent right-distributive strong bimonoid.
\end{lemma}
\begin{proof} By the definition of $+_\mathsf{ST}$ and $\times_\mathsf{ST}$,  the terms 0 and 1 are the additive and multiplicative neutral elements in $\sfST_\Sigma(X)$, respectively. Moreover, 0 is annihilating with respect to $\times_\mathsf{ST}$. Each of these properties can be described by a corresponding identity satisfied by  $\sfST_\Sigma(X)$. By Lemma \ref{identitiestofactoralgebras}, $\sfS(X)$ also satisfies these identities. Hence, $0_\sfS$ and $1_\sfS$ are the additive and multiplicative neutral elements in $\sfS(X)$, respectively, and $0_\sfS$ is annihilating with respect to  $\times_\sfS$.
Moreover, by Lemma \ref{lm:free-algebra-quotient} (with $\sfA=\sfST_\Sigma(X)$), the algebra $\sfS(X)$ satisfies the identities  $e_1 - e_5$. Identities $e_1 - e_3$, together with the mentioned properties of $0_\sfS$ and $1_\sfS$, ensure that  $\sfS(X)$  is a strong bimonoid, then identities $e_4$ and $e_5$ ensure that it is almost idempotent and right-distributive, respectively.
\end{proof}

\underline{Step 2:} Clearly,  $\sfS(X)$  is not multiplicatively locally finite, since for each  $x \in X$,
the multiplicative submonoid  $\langle\{[x]_E\}\rangle_{\{\times_{\sfS},1_\sfS\}}$  generated by  $\{[x]_E\}$  is infinite.
In order to construct a quotient algebra that is
multiplicatively locally finite we define the concept  of a large element in  $\rmS(X)$. Intuitively, all products $ p \times_{\sfS} (q \times_{\sfS} r)$ of elements $p, q, r \in \rmS(X) \setminus \{ 0_\sfS, 1_\sfS \}$ are large, and any element obtained from a large element by adding or multiplying it with any further element should also be large.
The formal definition is the following.

\begin{definition}\label{def:large-polynomial} \rm An element $p \in \rmS(X) \setminus \{ 0_\sfS, 1_\sfS \}$  is \emph{large},
if  there is a term $s\in \rmST(X)$ such that:
\begin{compactitem}
\item[(a)]  $p = [s]_{E}$  and 
\item[(b)]  $s$ has a subterm of the form $t_1\times(t_2\times t_3)$ or  $(t_1\times t_2)\times t_3$ for some $t_1,t_2,t_3 \in  \rmST(X)$.\hfill$\Box$
\end{compactitem}
\end{definition}

The next observation is obvious by the definition of a large element.

\begin{observation}\label{lm:p-q-r-large}\label{lm:p-generatum-large}\rm $\,$ Let  $p, q, r \in \rmS(X) \setminus \{ 0_\sfS, 1_\sfS \}$.
\begin{compactitem}
\item[(a)] Then  $p \times_{\sfS} (q \times_{\sfS} r)$  is large. 
\item[(b)] If  $p$  is large, then  $p +_{\sfS} q$, $p \times_{\sfS} q$,  and  $q \times_{\sfS} p$  are also large. 
\end{compactitem}
\end{observation}

Next we will, intuitively, ``identify'' all large elements. For this, we define a binary relation $\sim_\mathrm{la}$  on  $\rmS(X)$  as follows: for every $q, r \in \rmS(X)$, we let
$q \sim_\mathrm{la} r$  if and only if  $q = r$  or  both  $q$  and  $r$ are large.

\begin{lemma}\label{lm:simL-congruence}\rm The relation $\sim_\mathrm{la}$ is a congruence on $\sfS(X)$.
\end{lemma}
\begin{proof} It is obvious that  $\sim_\mathrm{la}$ is an equivalence relation.
Let $p, q, r \in \rmS(X)$. If $q \sim_\mathrm{la} r$, then by Observation \ref{lm:p-generatum-large}(b), we have
$p +_{\sfS} q \sim_\mathrm{la} p +_{\sfS} r$, $p \times_{\sfS} q \sim_\mathrm{la} p \times_{\sfS} r$, and  $q \times_{\sfS} p \sim_\mathrm{la} r \times_{\sfS} p$.
This proves that $\sim_\mathrm{la}$ is a congruence.
\end{proof}

Now we consider the quotient algebra 
\[ \sfS(X)/_{\sim_\mathrm{la}}=(\rmS(X)/_{\sim_\mathrm{la}}, +_{\sfS}/_{\sim_\mathrm{la}}, \times_{\sfS}/_{\sim_\mathrm{la}}, [0_\sfS]_{\sim_\mathrm{la}}, [1_\sfS]_{\sim_\mathrm{la}}) \]
of $\sfS(X)$ with respect to $\sim_\mathrm{la}$, 
and we abbreviate the above notation by $\sfM(X) =(\rmM(X),\oplus,\otimes,\0,\1)$. Moreover, for  each $q\in \rmS(X)$, we abbreviate $[q]_{\sim_\mathrm{la}}$ by $[q]_\mathrm{la}$.

The algebra $\sfM(X)$ is an almost idempotent  right-distributive strong bimonoid because $\sfS(X)$ is a strong bimonoid which satisfies the same conditions (cf. Lemma \ref{prop:SB-strong-bimonoid}),
and by Lemma \ref{identitiestofactoralgebras} all identities satisfied by $\sfS(X)$ are satisfied also by $\sfM(X)$. Now we can prove the first main result of this paper.

\begin{theorem}\label{thm:M-is-what-we-want} The strong bimonoid $\sfM(X) =(\rmM(X),\oplus,\otimes,\0,\1)$ is almost idempotent,  right-distributive, bi-locally finite, and not locally finite.
\end{theorem}
\begin{proof} By Observation \ref{obs:biloc-fin+right-disrt-loc-fin}(b),  $\sfM(X)$  is additively locally finite.  We show that it  is also multiplicatively locally finite as follows.
Fix any  $q \in \rmS(X) \setminus \{ 0_\sfS, 1_\sfS\}$, and let   $p', q', r' \in \rmS(X) \setminus \{ 0_\sfS, 1_\sfS\}$.
Then  $p' \times_{\sfS} (q' \times_{\sfS} r') \sim_\mathrm{la} q \times_{\sfS} (q \times_{\sfS} q)$ because, by Observation \ref{lm:p-q-r-large}(a), both products are large. 
It follows that if  $F$  is a finite subset of $\rmM(X)$,
then the multiplicative submonoid  $\langle F\rangle_{\{\otimes,\1\}}$  of $(\rmM(X), \otimes, \1)$  generated by  $F$ contains $F \cup \{ \1\}$,
all binary products of elements of  $F$  and, possibly, $[q]_\mathrm{la}\otimes ([q]_\mathrm{la}\otimes [q]_\mathrm{la})$.
Thus $\langle F\rangle_{\{\otimes,\1\}}$ is finite. Hence, $\sfM(X)$  is bi-locally finite. 

It remains to show that  $\sfM(X)$ is not locally finite. 
For this, choose an $x \in X$. We define, for each  $n \in \mathbb{N}$,  the term  $t_n \in \rmST(X)$  by induction as follows: $t_0 = x$  and  $t_{n+1} = x \times (1 + t_n)$. So, e.g.,  $t_1 = x \times (1 + x)$  and $t_2 = x \times (1 + t_1) =  x \times (1 + (x \times (1 + x)))$ 

Next, for each $n \in \mathbb{N}$, we consider the $\approx_{E}$-class which contains the term $t_n$, i.e., we
let  $p_n = [t_n]_{E}$. Then $p_n \in \rmS(X)$, and we have $p_0 = [x]_{E}$  and  $p_{n+1} = [x]_{E} \times_\sfS (1_\sfS +_\sfS p_n)$  for  each $n \in \mathbb{N}$.

Now we show that, for each  $n \in \mathbb{N}$,  $p_n$  is not large, i.e., there does not exist $t\in [t_n]_{E}$ such that $t$ has the property described in 
Definition \ref{def:large-polynomial}(b). For this, let $t\in [t_n]_{E}$.
By Lemma \ref{lm:approx-characterization} (with $\sfA=\sfST_\Sigma(X)$), we have $t_n \Leftrightarrow^*_{E} t$, i.e., $t_n$ can be transformed to $t$ in finitely many reduction steps  using in each step an identity in $E$ or its inverse.
However, due to the special shape of $t_n$, in each reduction step of the transformation only identity $e_2$ (commutativity of addition) can be used. 
Hence $t$ cannot be in the form described in Definition~\ref{def:large-polynomial}(b), which means that  $p_n$  is not large.

It also follows that the congruence class $p_n$ contains  $2^n$  elements, because there are $n$ occurrences of $+$ in $t_n$, and we can apply the identity $e_2$ in $n$ instances that are independent; thus, there are $2^n$ different elements in the class $p_n$. Consequently, we also obtain that for every $m,n \in \mathbb{N}$ with $m \neq n$, we have $p_m \neq p_n$.

Now for each  $n \in \mathbb{N}$, we let $a_n = [p_n]_{\mathrm{la}} \in \rmM(X)$. Then
$a_{n+1} = [[x]_{E}]_\mathrm{la} \otimes (\1 \oplus  a_n)$  for each  $n \in \mathbb{N}$,
so  $\{a_n \mid n\in \mathbb{N}\} \subseteq \langle \{ \1, a_0 \} \rangle_{\{\oplus, \otimes\}}$. Finally, let  $m, n \in \mathbb{N}$  with  $m \neq n$. Since  also both  $p_m$  and  $p_n$ are not large, we have $p_m \not\sim_\mathrm{la} p_n$, showing  $a_m \neq a_n$,
i.e., the set $\{a_n \mid n\in \mathbb{N}\}$ is infinite.
Consequently,  $\langle \{ \1, a_0 \} \rangle_{\{\oplus, \otimes\}}$  is infinite,
showing that  $\sfM(X)$  is not locally finite. 
\end{proof}

\section{Weighted tree automata}
\label{sect:wta-def}

\emph{In this section, let $\Sigma$ be an arbitrary ranked alphabet and  $\B=(B,\oplus,\otimes,\0,\1)$ be an arbitrary  strong bimonoid, if not stated otherwise.}

A \emph{$(\Sigma,\B)$-weighted tree automaton} (for short: $(\Sigma,\B)$-wta, or simply: wta) is a tuple $\cA=(Q,\delta,F)$ 
 where $Q$ is a finite non-empty set (\emph{states}),
 $\delta=(\delta_k\mid k\in\mathbb{N})$ is a family of mappings $\delta_k: Q^k\times \Sigma^{(k)}\times Q \to B$ (\emph{transition mappings})  where we consider $Q^k$ as set of words over $Q$ of length $k$, and 
 $F: Q \rightarrow B$ is a mapping (\emph{root weight vector}). We denote by $\mathrm{wts}(\cA)$ the set of all weights occurring in $\cA$, i.e., $\mathrm{wts}(\cA) = \bigcup_{k \in \mathbb{N}} \im(\delta_k) \cup \im(F)$.

 Let  $\cA=(Q,\delta,F)$ be a $(\Sigma,\B)$-wta. The {\em vector algebra of $\cA$} is the $\Sigma$-algebra $\V(\cA)=(B^Q,\delta_\cA)$ where, for every $k \in \mathbb{N}$, $\sigma \in \Sigma^{(k)}$, the $k$-ary operation $\delta_\cA(\sigma): B^Q \times \cdots \times B^Q \to B^Q$ is defined by 
\begin{equation}\label{eq:delta-A-definition}
\delta_\cA(\sigma)(v_1,\dots,v_k)_q 
  = \bigoplus_{q_1\cdots q_k \in Q^k} \Big(\bigotimes_{i\in[k]} (v_i)_{q_i}\Big) \otimes \delta_k(q_1\cdots q_k,\sigma,q)
  \end{equation}
  for every $v_1,\dots,v_k \in B^Q$ and $q \in Q$,
where, for each family $(b_i \mid i \in [k])$, the expression $\bigotimes_{i \in [k]} b_i$ stands for $b_1 \otimes \ldots \otimes b_k$.  Thus $\delta_\cA(\alpha)()_q = \delta_0(\varepsilon,\alpha,q)$  for each $\alpha \in \Sigma^{(0)}$, because $Q^0=\{\varepsilon\}$ and $\bigotimes_{i\in \emptyset} (v_i)_{q_i}=\1$.

Since the $\Sigma$-term algebra $\sfT_\Sigma$ is initial, there exists a unique $\Sigma$-algebra homomorphism from $\sfT_\Sigma$ to the vector algebra $\V(\cA)$; we denote it by $\h_\cA$. Then, for every $t = \sigma(t_1,\ldots,t_k)$ in $\T_\Sigma$ and $q\in Q$, we have
\begin{align*}
\h_\cA(\sigma(t_1,\ldots,t_k))_q &= \h_\cA(\theta_\Sigma(\sigma)(t_1,\ldots,t_k))_q = \delta_{\cA}(\sigma)
(\h_\cA(t_1),\ldots, \h_\cA(t_k))_q \\
&= \bigoplus_{q_1 \cdots q_k \in Q^k} \Big( \bigotimes_{i\in [k]} \h_\cA(t_i)_{q_i}\Big) \otimes \delta_k(q_1\cdots q_k,\sigma,q),
\end{align*}
where $\theta_\Sigma(\sigma)$ is the operation of the $\Sigma$-term algebra associated to $\sigma$; the second equality holds, because $\h_\cA$ is a $\Sigma$-algebra homomorphism. In particular, for each $\alpha \in \Sigma^{(0)}$, we have $\h_\cA(\alpha)_q = \delta_0(\varepsilon,\alpha,q)$.

The \emph{initial algebra semantics of $\cA$}, denoted by $\initialsem{\cA}$, is the weighted tree language $\initialsem{\cA}: \T_\Sigma \rightarrow B$  defined for every $t\in \T_\Sigma$ by  
\begin{equation*}
  \initialsem{\cA}(t) = \bigoplus_{q \in Q} \h_\cA(t)_q \otimes F_q\enspace. 
\end{equation*}

Next we recall the run semantics of the $(\Sigma,\B)$-wta $\cA$.
Let $t \in \T_\Sigma$. We define the \emph{set of positions of $t$}, denoted by $\pos(t)$, by structural induction as follows: (i) For each $t \in \Sigma^{(0)}$, we let $\pos(t) =\{\varepsilon\}$ and (ii) for every $k \in \mathbb{N}_+$, $\sigma \in \Sigma^{(k)}$, and $t_1,\ldots,t_k \in \T_\Sigma$, we let $\pos(\sigma(t_1,\ldots,t_k)) = \{\varepsilon\} \cup \bigcup_{i \in [k]} \{iw \mid w \in \pos(t_i)\}$. In particular, $\pos(t) \subseteq (\mathbb{N}_+)^*$.

Then, for every $t = \sigma(t_1,\ldots,t_k)$ in $\T_\Sigma$, a~\emph{run of $\cA$ on $t$} is a mapping $\rho: \pos(t) \rightarrow Q$. The \emph{set of all runs of $\cA$ on $t$} is denoted by $\R_\cA(t)$. Next we define the mapping $\wt_\cA: \mathrm{TR} \to B$ by structural induction on \(\mathrm{TR} = \{(t,\rho) \mid t \in \T_\Sigma, \rho \in \R_\cA(t)\}\), for every $t = \sigma(t_1,\ldots,t_k)$ in $\T_\Sigma$ and $\rho \in \R_\cA(t)$, as follows: 
 \begin{equation}\label{equ:weight-of-run}
\wt_\cA(t,\rho) = \Big( \bigotimes_{i\in [k]} \wt_\cA(t_i,\rho_i)\Big) \otimes \delta_k\big(\rho(1) \cdots \rho(k),\sigma,\rho(\varepsilon)\big) \enspace,
\end{equation}
where, for each $i \in [k]$, the run $\rho_i: \pos(t_i) \to Q$ of $\cA$ on $t_i$ is defined, for each $w \in \pos(t_i)$, by $\rho_i(w) = \rho(iw)$.

The {\it run semantics of $\cA$}, denoted by $\runsem{\cA}$, is the weighted tree language $\runsem{\cA}:~\T_\Sigma~\rightarrow~B$ such that,  for each $t \in \T_\Sigma$, we let
\begin{equation*}
  \runsem{\cA}(t) = \bigoplus_{\rho \in \R_\cA(t)}\wt_\cA(t,\rho) \otimes F_{\rho(\varepsilon)}\enspace. 
\end{equation*}

In general, the initial algebra semantics of $\cA$ is different from the run semantics of $\cA$, cf. e.g., \cite[Ex.~6.2.2-6.2.4]{fulvog26} and also Example~\ref{ex:Maps-wta}. However, the following equivalence is known.

\begin{theorem}{\rm (cf. \cite[Thm.~4.1]{rad10}, \cite[Lm. 4.1.13]{bor04b}, and \cite[Thm.~6.3.5]{fulvog26}; for weighted string automata cf. \cite[Lm.~4]{drostuvog10})} \label{thm:semiring-run=initial} Let $\Sigma$ be a ranked alphabet. Moreover, let $\B=(B,\oplus,\otimes,\0,\1)$ be a strong bimonoid. The following two statements are equivalent:
\begin{compactenum}
\item[(A)] If $\Sigma\not=\Sigma^{(0)}$, then $\B$ is right-distributive, and if $\Sigma\not= \Sigma^{(0)} \cup \Sigma^{(1)}$, then $\B$ is left-distributive.
\item[(B)] For each $(\Sigma,\B)$-weighted tree automaton $\cA$ we have $\runsem{\cA} = \initialsem{\cA}$.
\end{compactenum}
Hence, if $\B$ is a semiring, then $\initialsem{\cA}=\runsem{\cA}$ for each $(\Sigma,\B)$-weighted tree automaton $\cA$.
\end{theorem}

Essentially, a weighted automaton over words in $\Gamma^*$ (for an alphabet $\Gamma$) is a wta over the string ranked alphabet $\Gamma_e = \Gamma_e^{(0)} \cup \Gamma_e^{(1)}$ where $\Gamma_e^{(0)}=\{e\}$ for some $e \not\in \Gamma$, and $\Gamma_e^{(1)}= \Gamma$. The vector $(\delta(\varepsilon,e,q) \mid q \in Q)$ forms the initial weight vector. For a detailed explanation we refer to \cite[p.~324]{fulvog09new} and \cite[Sec.~4.1]{fulvog26}. We finish this section with an example of a weighted tree automaton over a ranked alphabet which contains a binary symbol (i.e., a symbol with rank 2) and over a right-distributive and not left-distributive strong bimonoid; for this weighted tree automaton the two semantics are different.

\begin{example}\rm \label{ex:Maps-wta} Let $\mathrm{Maps}(\mathbb{N})$ be the strong bimonoid of Example \ref{ex:maps-N}. Recall that it is right-distributive and not left-distributive.  We let $\Sigma = \{\sigma^{(2)}, \gamma^{(1)},\alpha^{(0)}\}$. We consider the $(\Sigma,\mathrm{Maps}(\mathbb{N}))$-wta $\cA=(Q,\delta,F)$ where $Q= \{q,p_1,p_2,p,q_f\}$, $F_{q_f}=\id$ and $F_q=F_{p_1}= F_{p_2} = F_p= \widetilde{0}$, and
    \begin{align*}
      \delta_0(\varepsilon,\alpha,q)=\mathrm{sq} \ \ \ \  \ \ & \delta_0(\varepsilon,\alpha,p_1)=\id  \ \ \ \ \ \ \delta_0(\varepsilon,\alpha,p_2)=\id  \\
      \delta_1(p_1,\gamma,p) = \id \ \ \ \  \ \ & \delta_1(p_2,\gamma,p) = \id  \ \ \ \ \ \ \delta_2(qp,\sigma,q_f) = \id
    \end{align*}
    and for each other transition the weight is $\widetilde{0}$. Let $\xi=\sigma(\alpha,\gamma(\alpha))$. Then we have:
  \begin{align*}
    \initialsem{\cA}(\xi)
    &= \h_\cA(\sigma(\alpha,\gamma(\alpha)))_{q_f} \circ \id
    = \h_\cA(\alpha)_q \circ \h_\cA(\gamma(\alpha))_p \circ \id\\
    &= \mathrm{sq} \circ (\h_\cA(\alpha)_{p_1} \circ \id + \h_\cA(\alpha)_{p_2} \circ \id)
    = \mathrm{sq} \circ (\id + \id) \enspace.
    \end{align*}
Let $\rho_1$ and $\rho_2$ be the following runs of $\cA$ on $\xi$:
$\rho_1(\varepsilon)=q_f$, $\rho_1(1)=q$, $\rho_1(2)=p$, and $\rho_1(21)=p_1$,
and $\rho_2$ is the same as $\rho_1$ except that $\rho_2(21)=p_2$.
Then
\begin{align*}
  \runsem{\cA}(\xi)= \wt_\cA(\xi,\rho_1) \circ \id + \wt_\cA(\xi,\rho_2) \circ \id
  = \mathrm{sq} \circ \id \circ \id \circ \id + \mathrm{sq} \circ \id \circ \id \circ \id
  = \mathrm{sq} + \mathrm{sq} \enspace.
\end{align*}
Since $\mathrm{sq} \circ (\id + \id) \ne \mathrm{sq} + \mathrm{sq}$, we have $\initialsem{\cA} \ne \runsem{\cA}$.
\hfill $\Box$
\end{example}

 %%%%%%%%%%%%%%%%%%%%%%%%%%%%%%%%%%%%%%%%%%%%%%%%

\section{Universality property of initial algebra semantics }\label{sect:wta-section}

In this section we will prove Main Theorem \ref{thm:main-result-on-Sigma}. That is, if the ranked alphabet $\Sigma$ contains a symbol with rank at least 2, then for each finite subset $A\subseteq B$, we can construct a $(\Sigma,\B)$-wta $\cA$ such that the image of its initial algebra semantics is equal to the closure of $A$, i.e., $\im(\initialsem{\cA}) = \langle A \rangle_{\{\oplus,\otimes,\0,\1\}}$ (cf. Theorem \ref{lm:closure-of-finite-set-i-recognizable}).
Weaker versions of this result were proved in \cite[Lm.~6.1]{rad10} (under the assumption that $\Sigma$ contains at least $|A\cup\{\0,\1\}|$ many nullary symbols and at least two binary symbols) and also in \cite[Lm.~12]{fulvog23} (there $\B$ is a bounded lattice and idempotency and commutativity are used in the proof).
After \cite{drofultepvog24} was published, Theorem \ref{lm:closure-of-finite-set-i-recognizable-stronger} was proved in \cite[Thm. 5.5.2]{fulvog26}.

Our proof of Theorem \ref{lm:closure-of-finite-set-i-recognizable-stronger} proceeds  as follows.
First, we prove a weaker version in which the existence of a binary symbol is required (cf.  Theorem \ref{lm:closure-of-finite-set-i-recognizable}). Then, by means of an easy padding lemma (cf. Lemma~\ref{lm:padding}) we lift Theorem \ref{lm:closure-of-finite-set-i-recognizable} to the stronger 
Theorem \ref{lm:closure-of-finite-set-i-recognizable-stronger}.

In the usual way, each value $a \in \langle A \rangle_{\{\oplus,\otimes,\0,\1\}}$ can be represented by an expression $t$, using the elements of $A \cup \{\0,\1\}$ as nullary symbols and $\oplus$ and $\otimes$ as binary symbols. The canonical mapping $\mathrm{eval}$ maps such an expression $t$ to its unique value in $B$.
Now let $\sigma \in \Sigma^{(2)}$ and $\alpha \in \Sigma^{(0)}$.
We will encode each $a_i \in A \cup \{\0,\1\}$ by a left-branching comb $f_i = \sigma(... \sigma(\alpha,\alpha)...,\alpha)$ with $i$ occurrences of $\sigma$. Also, we will encode $\oplus$ and $\otimes$ by the tree patterns $\sigma(\alpha,\sigma(.,.))$ and $\sigma(\sigma(.,.),\alpha)$, respectively. Let $g(t)$ be this encoding of an expression $t$. Then we construct $\cA$ such that $\initialsem{\cA}(g(t))=\mathrm{eval}(t)$. 

Formally, we assume that $A \cup \{\0,\1\} = \{a_1,\ldots,a_n\}$. We define the ranked alphabet $\Delta =\{\oplus^{(2)}, \otimes^{(2)}\} \cup \{a_1^{(0)},\ldots,a_n^{(0)}\}$ and the mapping $\mathrm{eval}: \T_\Delta \to \langle A \rangle_{\{\oplus,\otimes,\0,\1\}}$ by structural induction, for each $t \in \T_\Delta$, by:
\[
  \mathrm{eval}(t) = \begin{cases} a_i & \text{ if $t = a_i$ for some $i \in [n]$}\\
    \mathrm{eval}(t_1) \oplus \mathrm{eval}(t_2) & \text{ if $t = \oplus(t_1,t_2)$ for some $t_1,t_2 \in \T_\Delta$}\\
    \mathrm{eval}(t_1) \otimes \mathrm{eval}(t_2) & \text{ if $t = \otimes(t_1,t_2)$ for some $t_1,t_2 \in \T_\Delta$}\enspace.
    \end{cases}
  \]
  Clearly, $\mathrm{eval}$ is surjective. Thus, each value in $\langle A \rangle_{\{\oplus,\otimes,\0,\1\}}$ is represented by at least one tree in $\T_\Delta$.
  
Next we fix arbitrary $\alpha \in \Sigma^{(0)}$ and  $\sigma \in \Sigma^{(2)}$. For each $i \in [0,n]$ we define the $\Sigma$-tree $f_i$ by $f_0 = \alpha$ and $f_{i+1} = \sigma(f_i,\alpha)$ for each $i \in [n-1]$. Then, by structural induction, we define the mapping $g: \T_\Delta \to \T_\Sigma$, for each $t\in \T_\Delta$, as follows:
\begin{compactitem}
\item if $t=a_i$ for some $i\in[n]$, then $g(t)=f_i$,
\item if $t=\oplus(t_1,t_2)$, then $g(t)=\sigma(\alpha,\sigma(g(t_1),g(t_2)))$, and
\item  if $t=\otimes(t_1,t_2)$, then $g(t)= \sigma(\sigma(g(t_1),g(t_2)),\alpha)$.
\end{compactitem}
In fact, $g$ is a $(\Delta,\Sigma)$-tree homomorphism in the sense of \cite[Def.~3.62]{eng75-15},
and we can consider it as a coding. 

Then we construct the $(\Sigma,\B)$-wta $\cA$ such that $\initialsem{\cA}(g(t)) = \mathrm{eval}(t)$ for each $t \in \T_\Delta$. Thus, since   $\mathrm{eval}$ is surjective, each element of $\langle A \rangle_{\{\oplus,\otimes,\0,\1\}}$ occurs in the image of $\initialsem{\cA}$.

\begin{figure}[t]
  \centering
  
\begin{tikzpicture}[scale=0.8]
% nodes bottom
\tikzset{node distance=7em, scale=0.4, transform shape}
\node[state, rectangle] (b) {\huge $\alpha$};
  \node[state, right of=b] (q0) {\huge $q_0$};
\node[state, rectangle, right of=q0](g1-1) {\huge $\sigma$};
  \node[state, right of=g1-1] (q1) {\huge $q_1$};
\node[state, rectangle, right of=q1] (g1-2) {\huge $\sigma$};
  \node[state, right of=g1-2] (q2) {\huge $q_2$};
\node[state, right of=q2,opacity=0] (empty) {}; % empty node
\node[state, rectangle, right of=empty]  (g1-3) {\huge $\sigma$};
\node[state, right of=g1-3] (qn-1) {\huge $q_{n-1}$};

% nodes bottom plus one
\node[state, rectangle,above of=q0] (h0) {\huge $\sigma$};
\node[state, rectangle,above of=q1] (h1) {\huge $\sigma$};
\node[state, rectangle,above of=q2] (h2) {\huge $\sigma$};
\node[state, rectangle,above of=qn-1] (hn-1) {\huge $\sigma$};

% weights bottom plus one
\tikzset{node distance=3em}
\node[left of=h0]     {\LARGE $a_1$};
\node[left of=h1]     {\LARGE $a_2$};
\node[left of=h2]     {\LARGE $a_3$};
\node[left of=hn-1]     {\LARGE $a_n$};

% edges from bottom to bottom plus one
\draw (q0)    edge[->,>=stealth, out=300, in=350] (h0);
\draw (q0)    edge[->,>=stealth, out=90, in=270] (h0);

\draw (q0)    edge[->,>=stealth, out=290, in=350,looseness=1.8] (h1);
\draw (q1)    edge[->,>=stealth, out=90, in=270] (h1);

\draw (q0)    edge[->,>=stealth, out=280, in=350, looseness=1.5] (h2);
\draw (q2)    edge[->,>=stealth, out=90, in=270] (h2);

\draw (q0)    edge[->,>=stealth, out=270, in=350, looseness=1.3] (hn-1);
\draw (qn-1)    edge[->,>=stealth, out=90, in=270] (hn-1);

% ellipsis
\node (dots) at ($(q2.east)!0.5!(empty.east)$) {$\ldots$};

% weights
\tikzset{node distance=2em}
\node[below of=b]    (wb)    {\LARGE $\1$};
\node[below of=g1-1] (wg1-1) {\LARGE $\1$};
\node[below of=g1-2] (wg1-2) {\LARGE $\1$};
\node[below of=g1-3] (wg1-3) {\LARGE $\1$};
(
% edges ->
\draw (b)    edge[->,>=stealth] (q0);
\draw (q0)   edge[->,>=stealth, out=0, in=180] (g1-1);
\draw (q0)   edge[->,>=stealth, out=-20, in=210] (g1-1);
\draw (g1-1) edge[->,>=stealth] (q1);
\draw (q1)   edge[->,>=stealth, out=0, in=180] (g1-2);
\draw (q0)   edge[->,>=stealth, out=-30, in=210] (g1-2);
\draw (g1-2) edge[->,>=stealth] (q2);
\draw (empty)edge[->,>=stealth] (g1-3);
\draw (g1-3) edge[->,>=stealth] (qn-1);
\draw (q0)   edge[->,>=stealth, out=-40, in=210] (g1-3);

% state v
\node[state, above=5cm of g1-2] (v) {\huge $v$};
\draw (h0)    edge[->,>=stealth] (v);
\draw (h1)    edge[->,>=stealth] (v);
\draw (h2)    edge[->,>=stealth] (v);
\draw (hn-1)    edge[->,>=stealth] (v);
\node[below of=v, yshift=-3mm]     {\LARGE $\1$};

% upper part left
\node[state, rectangle,above=6cm of v,xshift=-6cm] (sigma1) {\huge $\sigma$};
\node[state,above=1cm of sigma1] (state1) {\huge $q_{\mathrm{one}}$};
\draw (state1) edge[->,>=stealth,out=200, in=250,looseness=1.7]  (sigma1);
\draw (state1) edge[->,>=stealth,out=340, in=290,looseness=1.6]  (sigma1);
\draw (sigma1) edge[->,>=stealth]  (state1);
\tikzset{node distance=3em}
\node[below of=sigma1]     {\LARGE $\1$};

\node[state, rectangle,left=2cm of state1] (sigmaleftleft) {\huge $\sigma$};
\draw (v) edge[->,>=stealth,out=140, in=270]  (sigmaleftleft);
\draw (state1) edge[->,>=stealth]  (sigmaleftleft);
\node[left of=sigmaleftleft]     {\LARGE $\1$};

\node[state, rectangle,right=2cm of state1] (sigmaleftright) {\huge $\sigma$};
\draw (v) edge[->,>=stealth,out=110, in=270]  (sigmaleftright);
\draw (state1) edge[->,>=stealth]  (sigmaleftright);
\node[right of=sigmaleftright]     {\LARGE $\1$};

\node[state, rectangle,above=1cm of state1] (alpha1) {\huge $\alpha$};
\draw (alpha1) edge[->,>=stealth] (state1);
\node[above of=alpha1,yshift=-2mm]     {\LARGE $\1$};

\node[state,above=1.5cm of alpha1] (oplus) {\huge $q_\oplus$};
\draw (sigmaleftleft) edge[->,>=stealth,out=90,in=210] (oplus);
\draw (sigmaleftright) edge[->,>=stealth,out=90,in=330] (oplus);

\node[state,rectangle,above=2cm of oplus] (sigmatopleft) {\huge $\sigma$};
\draw (sigmatopleft) edge[->,>=stealth,out=130,in=180, looseness=2.2] (v);
\node[state,left=2cm of oplus] (q0left) {\huge $q_0$};
\draw (q0left) edge[->,>=stealth] (sigmatopleft);
\draw (oplus) edge[->,>=stealth] (sigmatopleft);
\node[above of=sigmatopleft, xshift=0.1cm]     {\LARGE $\1$};

% upper part right
\node[state, rectangle,above=10cm of v,xshift=8cm] (sigma2) {\huge $\sigma$};
\node[right=3mm of sigma2] {\Large $\1$};
\node[state,above=1cm of sigma2] (otimes) {\huge $q_\otimes$};
\draw (v) edge[->,>=stealth,out=60,in=240] (sigma2);
\draw (v) edge[->,>=stealth,out=40,in=280] (sigma2);
\draw (sigma2) edge[->,>=stealth] (otimes);
\node[state, rectangle, above=2cm of otimes] (sigmatopright) {\huge $\sigma$}; 
\draw (otimes) edge[->,>=stealth] (sigmatopright);
\draw (sigmatopright) edge[->,>=stealth,out=50,in=10,looseness=2.1] (v);
\node[right=3mm of sigmatopright] {\Large $\1$};
\node[state,right=1cm of otimes] (q0right) {\huge $q_0$};
\draw (q0right) edge[->,>=stealth] (sigmatopright);
\end{tikzpicture}  

\vspace{-10mm}

  \caption{\label{fig:closure-by-initsem} The $(\Sigma,\B)$-wta $\cA$ of the proof of Theorem \ref{lm:closure-of-finite-set-i-recognizable}, where the three occurrences of the state $q_0$ have to be identified. }
\end{figure}
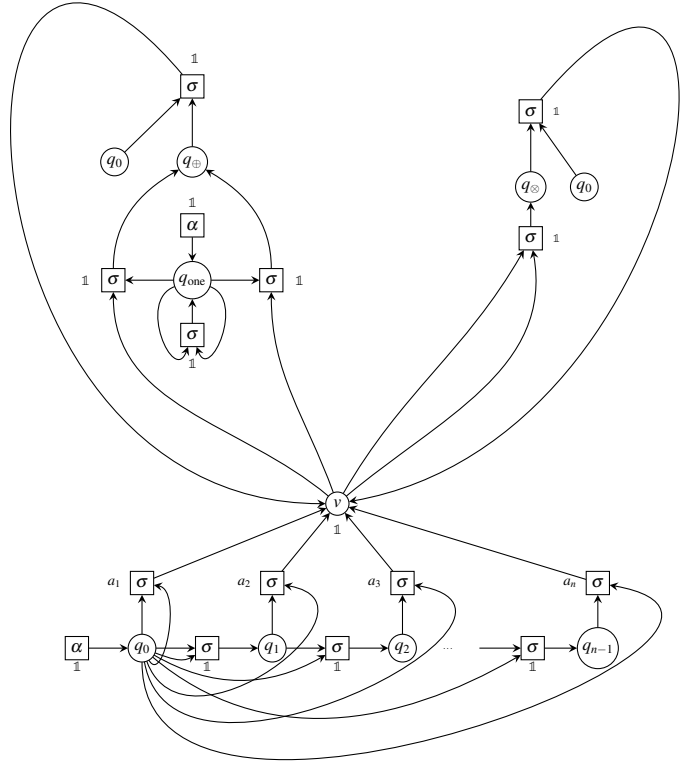

\begin{theorem}\label{lm:closure-of-finite-set-i-recognizable}  Let $\Sigma$ be a ranked alphabet which contains a binary symbol. Moreover, let $\B=(B,\oplus,\otimes,\0,\1)$ be a strong bimonoid and $A \subseteq B$  be a finite subset. Then we can construct a  $(\Sigma,\B)$-weighted tree automaton  $\cA$  such that  $\im(\initialsem{\cA}) = \langle A \rangle_{\{\oplus,\otimes,\0,\1\}}$. In particular, if  $\B$  is generated by  $A$, then we obtain  $\im(\initialsem{\cA}) = B$.
\end{theorem}

\begin{proof} Clearly, $\langle A \rangle_{\{\oplus,\otimes,\0,\1\}} = \langle A \cup \{\0,\1\} \rangle_{\{\oplus,\otimes\}}$. Let $a_1,\ldots,a_n$ be the elements of $A \cup \{\0,\1\}$. Let $\alpha \in \Sigma^{(0)}$ and $\sigma \in \Sigma^{(2)}$ be arbitrary elements. Moreover, the ranked alphabet $\Delta$ and the trees $f_1,\ldots,f_n \in \T_\Sigma$ are defined as above.
  Now we construct the $(\Sigma,\B)$-wta $\cA=(Q,\delta,F)$ as follows (cf. Figure \ref{fig:closure-by-initsem}).

\begin{compactitem}
\item $Q = \{v\} \cup \{q_0,\ldots,q_{n-1}\} \cup \{q_{\mathrm{one}},q_{\oplus}\} \cup \{q_{\otimes}\}$; the intention of the states is as follows:
  \begin{compactitem}
  \item $v$ is the ``main'' state with $\h_\cA(g(t))_v = \mathrm{eval}(t)$ for each $t \in \T_\Delta$ (cf. \eqref{eq:wta-computes-value-on-representations});
  \item the state  $q_i$ with $i \in [0,n-1]$ is used to recognize the tree $f_{i}$ with weight $\1$ (cf. \eqref{eq:usual-values-1}), and in combination with $v$, we will have $\h_\cA(\sigma(f_{i-1},\alpha))_v = a_i$;
  \item the states $q_{\oplus}$ and $q_{\otimes}$ are intermediate states such that $\h_\cA(\sigma(t_1,t_2))_{q_\oplus} = \mathrm{eval}(t_1) \oplus \mathrm{eval}(t_2)$ and  $\h_\cA(\sigma(t_1,t_2)))_{q_\otimes} = \mathrm{eval}(t_1) \otimes \mathrm{eval}(t_2)$, respectively; the state $q_{\mathrm{one}}$ supports $q_\oplus$;
    \item the switch from $q_\oplus$ to $v$ and from $q_\otimes$ to $v$ is triggered by the patterns $\sigma(\alpha,.)$ and $\sigma(.,\alpha)$, respectively,
    \end{compactitem}
\item $F_v=\mathbb{1}$ and, for each $q \in Q \setminus \{v\}$, we let $F_q=\mathbb{0}$, and
\item for each $q \in Q$, we define
  \(\delta_0(\varepsilon,\alpha,q) =\1\) if $q \in \{q_0,q_{\mathrm{one}}\}$, and     $\delta_0(\varepsilon,\alpha,q) =  \0$ otherwise,  and, for every $p,q,r \in Q$, we define
    \[
\delta_2(pq,\sigma,r) = 
\left\{
\begin{array}{ll}
\mathbb{1} & \text{if there exists $i \in [n-1]$ such that } pqr= q_{i-1}q_0q_i\\
  a_i & \text{if there exists $i \in [n]$ such that } pqr= q_{i-1}q_0v\\
  \1 & \text{if } pqr \in \{vq_{\mathrm{one}}q_{\oplus},\ q_{\mathrm{one}}vq_{\oplus},\ q_0q_{\oplus}v\} \\
  \1 & \text{if } pqr \in \{vvq_{\otimes},\ q_{\otimes}q_0v\}\\
    \1 & \text{if } pqr = q_{\mathrm{one}}q_{\mathrm{one}}q_{\mathrm{one}}\\
  \0 & \text{otherwise} \enspace,
\end{array}
\right.
\]
\item for every $k\in \mathbb{N}$, $\eta \in \Sigma^{(k)}$ with $\sigma\ne \eta\ne \alpha$, and $p,p_1,\ldots,p_k\in Q$, we define $\delta_k(p_1\ldots p_k,\eta,p)=\0$.
\end{compactitem}

Next we prove that $\im(\initialsem{\cA}) = \langle A \rangle_{\{\oplus,\otimes,\0,\1\}}$. Since $\mathrm{wts}(\cA) = A \cup \{\0,\1\}$, we have $\im(\initialsem{\cA}) \subseteq \langle A \rangle_{\{\oplus,\otimes,\0,\1\}}$. Thus it remains to prove $\langle A \rangle_{\{\oplus,\otimes,\0,\1\}}\subseteq \im(\initialsem{\cA})$.

For this purpose, we need some auxiliary statements, which are easy to see.
\begin{align}
&  \text{For each $q \in \{v,q_\oplus,q_\otimes\} \cup \{q_1,\ldots,q_{n-1}\}$, we have $\h_\cA(\alpha)_q=\0$\enspace.} \label{eq:usual-values-3} \\
& \text{For each $i \in [0,n]$, we have $\h_\cA(f_{i})_{q_\otimes}=\0$ \enspace.}
 \label{eq:usual-values-1.5} \\
& \text{For every $s_1,s_2 \in \T_\Sigma$, we have $\h_\cA(\sigma(s_1,s_2))_{q_0} = \0$\enspace.} \label{eq:usual-values-4} \\
 & \text{For every $s_1,s_2,s_3 \in \T_\Sigma$ and $i \in [0,n-1]$, we have $\h_\cA(\sigma(s_1,\sigma(s_2,s_3)))_{q_i} = \0$ \enspace.} \label{eq:usual-values-5} \\
 & \text{For each $s \in \T_{\{\sigma,\alpha\}}$, we have $\h_\cA(s)_{q_{\mathrm{one}}}=\1$\enspace.} \label{eq:usual-values-6}
 \end{align}
 
First, by bounded induction on $i$, we prove the following:
\begin{equation}
  \text{For each $i,j \in [0,n-1]$ with $i \le j$, we have $\h_\cA(f_i)_{q_{j}}=\h_\cA(f_0)_{q_{j-i}}$ \enspace.} \label{eq:usual-values-1.1}
\end{equation}
The case $i=0$ is trivial. For the induction step, let $i+1 \le j$. We can calculate as follows:
\begin{align*}
\h_\cA(f_{i+1})_{q_{j}} &= \h_\cA(\sigma(f_i,\alpha))_{q_j} = \h_\cA(f_i)_{q_{j-1}} \otimes \h_\cA(\alpha)_{q_0} \otimes \delta_2(q_{j-1}q_0,\sigma,q_j) = \h_\cA(f_i)_{q_{j-1}} = \h_\cA(f_0)_{q_{j-(i+1)}} \enspace,
\end{align*}
where the last equality holds by induction hypothesis. This proves \eqref{eq:usual-values-1.1}.

Second, by bounded induction on $i$, we prove that:
\begin{equation}
  \text{For each $i,j \in [0,n-1]$ with $i \ge j$, we have $\h_\cA(f_i)_{q_{j}}=\h_\cA(f_{i-j})_{q_0}$ \enspace.} \label{eq:usual-values-1.11}
\end{equation}
The case $i=0$ is trivial again. For the induction step, let $i+1 \ge j$. As above we obtain $\h_\cA(f_{i+1})_{q_{j}} =  \h_\cA(f_i)_{q_{j-1}}$. Then, by induction hypothesis, we obtain that $\h_\cA(f_i)_{q_{j-1}} = \h_\cA(f_{(i+1)-j})_{q_0}$. This  proves \eqref{eq:usual-values-1.11}.

Since,
\begin{compactitem}
\item for each $i < j$, Equality \eqref{eq:usual-values-1.1} implies that $\h_\cA(f_i)_{q_{j}} = \h_\cA(f_0)_{q_{j-i}} = \h_\cA(\alpha)_{q_{j-i}} = \0$ (because $q_{j-i} \not= q_0$), 
\item for each $i > j$, Equality \eqref{eq:usual-values-1.11} implies that $\h_\cA(f_i)_{q_{j}}=\h_\cA(f_{i-j})_{q_0}= \0$ (because $f_{i-j}\not= \alpha$), and
\item for $i=j$, Equalities \eqref{eq:usual-values-1.1} and \eqref{eq:usual-values-1.11} imply that $\h_\cA(f_i)_{q_{j}}=\h_\cA(f_0)_{q_{j-i}} = \h_\cA(f_{i-j})_{q_0}= \h_\cA(f_0)_{q_0} = \h_\cA(\alpha)_{q_0}= \1$,
\end{compactitem}
the following two statements hold:
\begin{align}
& \text{For each $i \in [0,n-1]$, we have $\h_\cA(f_i)_{q_{i}}  = \1$ \enspace.} \label{eq:usual-values-1}\\
& \text{For each $i,j \in [0,n-1]$ with $j \not= i$, we have $\h_\cA(f_i)_{q_{j}}=\0$ \enspace.} \label{eq:usual-values-1.2}
\end{align}

This finishes the proofs of auxiliary statements.
Then, by using these auxiliary statements and by structural induction, we can prove that:
\begin{equation}\label{eq:wta-computes-value-on-representations}
  \text{For each $t \in \T_\Delta$, we have $\h_\cA(g(t))_v = \mathrm{eval}(t)$} \enspace,
  \end{equation}
 where the mappings $\mathrm{eval}: \T_\Delta \to \langle A \rangle_{\{\oplus,\otimes,\0,\1\}}$ 
 and $g:\T_\Delta \to \T_\Sigma $ are defined before this theorem. %(for the proof, see Appendix).

Now let $a \in \langle A\rangle_{\{\oplus,\otimes,\0,\1\}}$. Since  $\mathrm{eval}$ is surjective, there exists $t \in \T_\Delta$ such that $\mathrm{eval}(t) = a$.  Then\\
\hspace*{20mm}
$\initialsem{\cA}(g(t)) = \bigoplus_{p \in Q} \h_\cA(g(t))_p \otimes F_p = \h_\cA(g(t))_v = \mathrm{eval}(t) = a$\\
where the last but one equality follows from   \eqref{eq:wta-computes-value-on-representations}.
Hence $\langle A\rangle_{\{\oplus,\otimes,\0,\1\}} \subseteq \im(\initialsem{\cA})$.
\end{proof}

We note that the proof of Theorem \ref{lm:closure-of-finite-set-i-recognizable}  is effective, even if the strong bimonoid  $\B$  is not given effectively:
Given the ranked alphabet  $\Sigma$  and the subset  $A$  generating  $\B$  as input data, the proof gives 
the construction of the requested wta  $\cA$  satisfying   
$\im(\initialsem{\cA}) = \langle A\rangle_{\{\oplus,\otimes,\0,\1\}}$.

Next we state a kind of padding lemma. Given a $(\Sigma,\B)$-wta $\cA$, it allows to extend the rank of symbols in $\Sigma$ (leading to the ranked alphabet $\Sigma'$) and it constructs a $(\Sigma',\B)$-wta $\cB$ which behaves on $\Sigma'$-input trees in the same way as $\cA$ behaves on the original $\Sigma$-trees.

Formally, let $(\Sigma,\rk)$ and $(\Sigma,\rk_e)$ be two ranked alphabets over the same set $\Sigma$. We say that  $(\Sigma,\rk_e)$ is an \emph{extension} of  $(\Sigma,\rk)$ if, for every $k \in \mathbb{N}$ and $\sigma \in \Sigma$, the relation $\rk(\sigma) \le \rk_e(\sigma)$ holds and, moreover, $\rk(\sigma)=0$ implies $\rk_e(\sigma)=0$. 
We denote the difference $\rk_e(\sigma) - \rk(\sigma)$ by $\e(\sigma)$. We use the standard definition of tree homomorphism as defined in \cite{gecste84,eng75-15}. 

Let $(\Sigma,\rk)$ and $(\Sigma,\rk_e)$ be as above. Let $\alpha$ be an arbitrary symbol in  $\Sigma^{(0)}$. 
We define the $((\Sigma,\rk),(\Sigma,\rk_e))$-tree homomorphism $g= (g_k \mid k \in \mathbb{N})$ as follows.
  For every $\sigma \in \Sigma$ and $k\in\mathbb{N}$ such that $\rk(\sigma)=k$, we let
  $g_k(\sigma) = \sigma(z_1,\ldots,z_k,\alpha,\ldots,\alpha)$ with $\e(\sigma)$ occurrences of $\alpha$. This tree homomorphism might be called a  padding.
  Clearly, for each $\xi \in \T_{(\Sigma,\rk)}$, we have $\pos(\xi)\subseteq \pos(g(\xi))$, and for each
  $w\in \pos(g(\xi))$, we have
  \begin{align}\label{eq:how-g-works}
  g(\xi)(w)= \begin{cases}
   \xi(w) & \text{ if $w\in \pos(\xi)$} \\
   \alpha & \text{ otherwise.}
  \end{cases}
  \end{align}

\begin{lemma}\rm \label{lm:padding} \cite[Lm.~5.4.1]{fulvog26} Let $(\Sigma,\rk)$ and $(\Sigma,\rk_e)$ be two ranked alphabets such that $(\Sigma,\rk_e)$ is an extension of $(\Sigma,\rk)$ and $\rk_e \ne \rk$. Let $\alpha$ be an arbitrary symbol in  $\Sigma^{(0)}$ and let $g$ be the tree homomorphism defined above. Moreover, let $\cA$ be a $((\Sigma,\rk),\B)$-wta. Then we can construct a $((\Sigma,\rk_e),\B)$-wta $\cB$ such that the following two statements hold.
  \begin{compactenum}
  \item[(1)] For each $\xi \in \T_{(\Sigma,\rk)}$ we have $\initialsem{\cB}(g(\xi)) = \initialsem{\cA}(\xi)$  and
    for each $\zeta \in \T_{(\Sigma,\rk_e)}\setminus g(\T_{(\Sigma,\rk)})$ \\ we have $\initialsem{\cB}(\zeta) = \0$. 
  \item[(2)] $\im(\initialsem{\cB}) = \im(\initialsem{\cA}) \cup\{\0\}$.
    \end{compactenum}
    \end{lemma}

\begin{proof}   We construct the $((\Sigma,\rk_e),\B)$-wta $\cB =(Q',\delta',F')$ by letting $Q' = Q \cup \{q_\alpha\}$ where $q_\alpha$ is a new state (i.e., $q_\alpha \not\in Q$), $F'_q = F_q$ for each $q \in Q$ and $F'_{q_\alpha}=\0$. Moreover, for every $\sigma \in \Sigma$ and $q_1',\ldots,q_{\rk_e(\sigma)}',q' \in Q'$, using $k$ as abbreviation for $\rk(\sigma)$, we let
  \begin{align*}
    &(\delta')_{\rk_e(\sigma)}(q_1' \cdots q_{\rk_e(\sigma)}',\sigma,q') =\\
    & \ \ \begin{cases}
      \delta_k(q_1'\cdots q_k',\sigma,q') & \text{ if $q_1',\ldots,q_k' \in Q$, $q_{k+1}' \cdots q_{\rk_e(\sigma)}' = {q_\alpha}^{\e(\sigma)}$, and $q' \in Q$}\\
      \1 & \text{ if $k=0$, $\sigma=\alpha$, and  $q' = q_\alpha$}\\
      \0 & \text{ otherwise}  \enspace.
      \end{cases}
    \end{align*}
    The proof of (1) is omitted. The proof of (2) easily follows from (1).
    \end{proof}

\begin{theorem}\label{lm:closure-of-finite-set-i-recognizable-stronger}  Let $\Sigma$ be a ranked alphabet which contains a symbol with rank at least 2, i.e., $\Sigma \ne \Sigma^{(0)} \cup \Sigma^{(1)}$. Moreover, let $\B=(B,\oplus,\otimes,\0,\1)$ be a strong bimonoid and $A \subseteq B$  be a finite subset. Then we can construct a  $(\Sigma,\B)$-weighted tree automaton  $\cA$  such that  $\im(\initialsem{\cA}) = \langle A \rangle_{\{\oplus,\otimes,\0,\1\}}$. In particular, if  $\B$  is generated by  $A$, then we obtain  $\im(\initialsem{\cA}) = B$.
\end{theorem}

\begin{proof} Let $(\Sigma,\rk)$ be a ranked alphabet with $\Sigma \ne \Sigma^{(0)} \cup \Sigma^{(1)}$. Hence there exists a symbol $\sigma \in \Sigma$ with $\rk(\sigma) \geq 2$.
If $\rk(\sigma) = 2$, then the claim is immediate by Theorem~\ref{lm:closure-of-finite-set-i-recognizable}. Therefore assume that  $\rk(\sigma) > 2$.
 We define the ranked alphabet $(\Sigma,\rk')$ such that 
$\rk'(\sigma) = 2$  and 
  $\rk'(\omega) = \rk(\omega)$ for each $\omega \in \Sigma \setminus \{\sigma\}$. 
 Then  $(\Sigma,\rk)$ is an extension of $(\Sigma,\rk')$. 
Now let  $A \subseteq B$  be a finite subset.
   First, since $(\Sigma,\rk')$ contains the binary symbol $\sigma$,  by Theorem~\ref{lm:closure-of-finite-set-i-recognizable}, we can construct a $((\Sigma,\rk'),\B)$-wta $\cA'$ such that $\im(\initialsem{\cA'}) = \langle A \rangle_{\{\oplus,\otimes,\0,\1\}}$.
  Second, by applying Lemma~\ref{lm:padding} to $(\Sigma,\rk')$, $(\Sigma,\rk)$, and $\cA'$, we obtain the $((\Sigma,\rk),\B)$-wta $\cA$. Since $\0 \in \langle A \rangle_{\{\oplus,\otimes,\0,\1\}}$, by Lemma~\ref{lm:padding}(2) we have $\im(\initialsem{\cA})= \langle A \rangle_{\{\oplus,\otimes,\0,\1\}}$. 
\end{proof}

 As an immediate consequence of Theorem~\ref{lm:closure-of-finite-set-i-recognizable-stronger}, we obtain the following result.

  \begin{corollary} \label{lm:weak-loc-fin-not-loc-fin-monadic-is-weaker}\rm Let $\Sigma$ be a ranked alphabet and $\B$ a strong bimonoid.  If $\Sigma$ contains a symbol with rank at least  2 and $\B$ is not locally finite, then there exists a  $(\Sigma,\B)$-weighted tree automaton $\cA$ such that the mapping $\initialsem{\cA}$ has an infinite image.
       \end{corollary}

Now the main consequence of our results follows from Theorem \ref{thm:M-is-what-we-want} and Corollary \ref{lm:weak-loc-fin-not-loc-fin-monadic-is-weaker}.

\begin{theorem}\label{cor:second-main} There exists a bi-locally finite right-distributive strong bimonoid $\B$ such that for each ranked alphabet $\Sigma$ which contains a symbol with rank at least 2, there exists a $(\Sigma,\B)$-weighted tree automaton $\cA$
for which $\im(\initialsem{\cA})$ is an infinite set.
\end{theorem}

Due to the equality of initial algebra semantics and run semantics of wta over semirings we obtain the following consequence of Theorem~\ref{thm:semiring-run=initial} and Theorem~\ref{lm:closure-of-finite-set-i-recognizable-stronger}.

\begin{corollary}\label{cor:second-main-semirings} \rm Let $\Sigma$ be an arbitrary ranked alphabet which contains a symbol with rank at least 2, and let $\B =(B,\oplus,\otimes,\0,\1)$ be any finitely generated semiring. Then there exists a weighted tree automaton $\cA$ over $\Sigma$ and $\B$ such
  that $\im(\runsem{\cA}) = B$.
  \end{corollary}

%%%%%%%%%%%%%%%%%%%%%%%%%%%%%%%%%%%%%%%%%%%%%%%%%%%%%%%%%%%%%%%%%%%%%%%%%%%%%%%%%%%%%%

\section{Further research}\label{sect:further-research}

The following nice characterizations of the finite image property of the initial algebra semantics resp. of the run semantics
are known.

\begin{theorem}\label{thm:loc-finite-rec-step-function} (cf. \cite[Thms.~16.1.6, ~16.2.7]{fulvog26}) Let  $\B$ be a strong bimonoid. 
\begin{compactenum}
\item[(A)]  $\B$ is locally finite if and only if for each ranked alphabet $\Sigma$ and for each $(\Sigma,\B)$-wta $\cA$, the set $\im(\initialsem{\cA})$  is finite.
\item[(B)]  $\B$ is bi-locally finite if and only if for each ranked alphabet $\Sigma$ and for each $(\Sigma,\B)$-wta $\cA$, the set $\im(\runsem{\cA})$  is finite.
\end{compactenum}
\end{theorem}

The following question arises for weighted string automata (phrased in our present terminology):

Characterize the strong bimonoids  $\B$  such that
\begin{equation}\label{equ:*}
\text{for each string ranked alphabet $\Sigma$ and for each $(\Sigma,\B)$-wta $\cA$, the set $\im(\initialsem{\cA})$  is finite.}
\end{equation}

In \cite{drostuvog10} weakly locally finite strong bimonoids were defined
 (cf.  also \cite{drofultepvog24}). It is known that \eqref{equ:*} holds true if  $\B$  is weakly locally finite \cite[Lm.~18]{drostuvog10}, cf. \cite[Lm.~16.1.1]{fulvog26}; if $\B$ is right-distributive, then it is bi-locally finite iff it is weakly locally finite (cf. \cite[Rem.~17]{drostuvog10}). Moreover, \eqref{equ:*} implies that  $\B$  is bi-locally finite (but not conversely), see \cite[Lm.~12, Ex.~25]{drostuvog10}.
We can show that \eqref{equ:*} holds true if and only if all matrix magmas  $(B^{n \times n},\cdot)$ $(n \in \mathbb{N})$ of  $n \times n$-matrices over  $\B$  are locally finite with respect to multiplication from the right.
However, a characterization entirely in terms of the structure of  $\B$  is open.
We conjecture that \eqref{equ:*} does not imply that  $\B$  is weakly locally finite.

\nocite{*}
\bibliographystyle{eptcs}
\bibliography{afl2026.bib}

\end{document}